\documentclass[journal]{IEEEtran}
\usepackage{amsmath,graphicx,cite,bm,amssymb,amsthm,enumerate,epsfig,psfrag,cases,mathtools}
\renewcommand{\(}{\left(}
\renewcommand{\)}{\right)}
\renewcommand{\[}{\left[}
\renewcommand{\]}{\right]}

\renewcommand{\j}{\mathbf{j}}

\renewcommand{\S}{\mathbf{S}}

\newcommand{\y}{\mathbf{y}}

\newcommand{\E}{\mathbf{E}}

\renewcommand{\r}{\mathbf{r}}

\newcommand{\Z}{\mathbf{Z}}

\newcommand{\x}{\mathbf{x}}

\newcommand{\I}{\mathbf{I}}
\newcommand{\J}{\mathbf{J}}
\newcommand{\C}{\mathbf{C}}
\renewcommand{\P}{\mathbf{P}}

\newcommand{\A}{\mathbf{A}}

\newcommand{\U}{\mathbf{U}}
\newcommand{\M}{\mathbf{M}}

\newcommand{\Q}{\mathbf{Q}}
\newcommand{\K}{\mathbf{K}}
\newcommand{\q}{\mathbf{q}}

\newcommand{\X}{\mathbf{X}}
\newcommand{\Y}{\mathbf{Y}}
\newcommand{\B}{\mathbf{B}}

\newcommand{\e}{\mathbf{e}}

\newcommand{\Tr}[1]{{\rm{Tr}}\left(#1\right)}

\newcommand{\End}[1]{{\rm{End}}}

\renewcommand{\log}[1]{{\rm{log}}#1}

\newcommand{\vecc}[1]{{\rm{vec}}\(#1\)}

\newcommand{\rank}[1]{{\rm{rank}}\(#1\)}
\newcommand{\sspan}[1]{\langle #1\rangle }

\newtheorem{lemma}{Lemma}
\newtheorem{definition}{Definition}
\newtheorem{theorem}{Theorem}

\newtheorem{corollary}{Corollary}

\newcommand{\norm}[1]{\left\lVert#1\right\rVert}

\usepackage{fontenc}
\usepackage{relsize}
\usepackage{inputenc,enumitem}
\usepackage[square,sort,compress,comma,numbers]{natbib}
\usepackage{MnSymbol}
\usepackage{epstopdf,pst-node}
\usepackage{tikz}
\usetikzlibrary{arrows,matrix,positioning}

\usepackage{faktor}

\usepackage{mathtools}

\begin{document}
\title{Group Symmetric Robust Covariance Estimation}

\author{Ilya Soloveychik, Dmitry Trushin and Ami Wiesel, \\ The Hebrew University of Jerusalem, Israel
\thanks{This work was supported by the Intel Collaboration Research Institute for Computational Intelligence, the Kaete Klausner Scholarship and ISF Grant 786/11.}
\thanks{The results were partially presented at the $1$-st IEEE Global Conference on Signal and Information Processing, December, 3-5, 2013, Austin, Texas, USA and at the $40$-th IEEE International Conference on Acoustics, Speech, and Signal Processing, April, 19-24, 2015, Brisbane, Australia.}
}
\maketitle

\begin{abstract}
In this paper we consider Tyler's robust covariance M-estimator under group symmetry constraints. We assume that the covariance matrix is invariant to the conjugation action of a unitary matrix group, referred to as group symmetry. Examples of group symmetric structures include circulant, perHermitian and proper quaternion matrices. We introduce a group symmetric version of Tyler's estimator (STyler) and provide an iterative fixed point algorithm to compute it. The classical results claim that at least $n=p+1$ sample points in general position are necessary to ensure the existence and uniqueness of Tyler's estimator, where $p$ is the ambient dimension. We show that the STyler requires significantly less samples. In some groups even two samples are enough to guarantee its existence and uniqueness. In addition, in the case of elliptical populations, we provide high probability bounds on the error of the STyler. These too, quantify the advantage of exploiting the symmetry structure. Finally, these theoretical results are supported by numerical simulations.
\end{abstract}


\section{Introduction}
Covariance matrix estimation is a fundamental problem in the field of statistical signal processing. Many algorithms for hypothesis testing, inference, denoising and prediction rely on accurate estimation of the second order statistic. The problem is especially challenging when the available data is high dimensional and non-Gaussian. Such settings are typical in many applications including speech, radar, wireless communication, finance and more. These led to a growing interest in both robust and structured covariance estimation. 

In robust statistics the non-Gaussian data is usually modeled using heavy-tailed distributions and outlier contamination \cite{huber1964robust}. These approaches have been found useful in different fields of statistical signal processing. Examples include $K$-distributed populations in the area of radar detection \cite{billingsley1993ground,watts1985radar}, Weibull distributions in biostatistics and radar detection \cite{farina1987radar}. An elegant multivariate model is provided by elliptically and Generalized Elliptically (GE) distributed random vectors \cite{frahm2004generalized}. These encompass a large number of non-Gaussian distributions including Gaussian, generalized Gaussian and compound Gaussian processes. Elliptical populations have been used in various problems such as bandlimited speech processing \cite{rupp1994analysis}, radar clutter echoes \cite{conte1994performance, gini1997sub}, wireless radio fading propagation problems \cite{abdi2003expected, yao2004unified}, anomaly detection in wireless sensor networks \cite{chen2011robust}, antenna array processing \cite{ollila2003robust}, radar detection \cite{abramovich2007time, ollila2012complex, bandiera2010knowledge, pascal2008covariance} and financial engineering  \cite{belkacem2000capm}.
A prominent robust covariance estimator motivated by this models is Maronna's M-estimator, \cite{maronna1976robust}. Later, Tyler developed a closely related distribution free M-estimator in \cite{tyler1987distribution}, which has become very popular and was extended to the complex case in \cite{pascal2008covariance}. Given $n$ i.i.d. (independent and identically distributed) measurements $\x_i \in \mathbb{C}^p,\;i=1,\dots,n$, Tyler's covariance matrix estimator is defined as the solution to the fixed point equation
\begin{equation}
\widehat{\bm\Theta} = \frac{p}{n}\sum_{i=1}^n \frac{\x_i\x_i^H}{\x_i^H\widehat{\bm\Theta}^{-1}\x_i}.
\label{tylerequ}
\end{equation}
The conditions under which Tyler's estimator exists and is unique have been thoroughly investigated  \cite{maronna1976robust, tyler1987distribution, pascal2008covariance, zhang2012robust}. One of the forms of sufficient conditions states that Tyler's estimator exists and is unique if $n > p$ and the data vectors lie in general position (i.e. any subset of $k\leqslant p$ vectors is linearly independent), which holds a.s. (almost surely) in the elliptical populations mentioned above. When this condition holds, the estimator can be easily computed using a simple fixed point iteration, \cite{tyler1987distribution, pascal2008covariance, zhang2012robust}. Recently, it was shown that the estimator is actually the solution to a $g$-convex maximum likelihood estimator (MLE) in Angular models, and that the iteration is a simple descent algorithm \cite{rapcsak1991geodesic, wiesel2012geodesic, wiesel2012unified, wiesel2011regularized}. In these models, high probability bounds were derived, \cite{soloveychik2014non}. They state that the Frobenius norm squared error of the matrix $\widehat{\bm\Theta}^{-1}$ decays like $p^2/n$ as long as $n>p$.

The above results all demonstrate the success of Tyler's estimator when $n \gg p$. On the other hand, in high dimensional cases where $p$ is large compared to $n$, the performance degrades, and when $n < p$, the estimator does not even exist. This is in particular the case in financial data analysis where few stationary monthly observations of numerous stock indexes are used to estimate the joint covariance matrix of the stock returns \cite{laloux2000random, ledoit2003improved} and bioinformatics where clustering of genes is obtained based on gene sequences sampled from a small population \cite{schafer2005shrinkage}. Additional applications include computational immunology where correlations among mutations in viral strains are estimated from sampled viral sequences and used as a basis of novel vaccine design \cite{dahirel2011coordinate, quadeer2013statistical}, psychology where the covariance matrix of multiple psychological traits is estimated from data collected on a group of tested individuals \cite{steiger1980tests} and electrical engineering at large where signal samples extracted from a possibly short time window are used to retrieve parameters of the signal \cite{scharf1991statistical}. This led to a body of literature on regularized and structured Tyler's estimator in sample starved scenarios. One of the common ways of introducing the prior information into covariance estimation is shrinkage towards a known matrix. The robust analogs of the shrunk sample covariance matrix \cite{ledoit2003improved} were developed in \cite{abramovich2007diagonally, chen2011robust, wiesel2012unified}. Another popular approach is to assume the true covariance matrix to possess linear structure \cite{snyder1989use, abramovich2007time, wiesel2013time, dembo1989embedding, cai2013optimal, bickel2008regularized}. Probably the most popular structure is the Toeplitz model \cite{snyder1989use, abramovich2007time, wiesel2013time}, closely related to it are circulant matrices \cite{dembo1989embedding, cai2013optimal}. In other settings the number of parameters can be reduced by assuming that the covariance matrix is sparse. A popular sparse model is the banded covariance, which is associated with time-varying moving average models \cite{bickel2008regularized,wiesel2013time}. Recently, a heuristic based on the method of moments was proposed to incorporate convex constraints into Tyler's estimator, \cite{soloveychik2014tyler}.

Our work was motivated by the paper \cite{shah2012group} which considered group symmetry structures in Gaussian populations. Given a finite group of unitary matrices $\mathcal{G}$, a covariance matrix $\bm\Theta$ is $\mathcal{G}$-invariant if it satisfies
\begin{equation}
\bm\Theta = \K^H \bm\Theta \K,\quad\forall \K \in \mathcal{G}.
\end{equation}
These structures are ubiquitous in statistical signal processing. Examples of group symmetric classes include circulant, perHermitian, and proper quaternion matrices. As expected, $\mathcal{G}$-invariance reduces the number of degrees of freedom in $\bm\Theta$. Group representation theory quantifies this reduction via two effects: (1) block sparsity and (2) block replication. Rigorously, any $\mathcal{G}$-invariant structure is equivalent (up to a change of basis) to a block diagonal matrix with, possibly, multiple identical blocks. By exploiting this property, \cite{shah2012group} showed that the Gaussian covariance MLE can be significantly outperformed. 

The goal of this paper is to exploit group symmetry in non-Gaussian distributions. We derive and analyze a symmetric version of Tyler's estimator, named STyler. We define the estimator by applying the original definition on synthetically $\mathcal{G}$-rotated samples, and obtain a $\mathcal{G}$-invariant estimator
\begin{equation}
\widehat{\bm\Theta}^\mathcal{G} = \frac{p}{n|\mathcal{G}|}\sum_{i=1}^n \sum_{\K \in \mathcal{G}} \frac{(\K\x_i) (\K\x_i)^H}{(\K\x_i)^H\[\widehat{\bm\Theta}^\mathcal{G}\]^{-1}\K\x_i}.
\label{def_ste_fp_intro}
\end{equation}
We derive conditions for the existence and uniqueness of the STyler, which are are directly related to the intrinsic degrees of freedom associated with the sparsity and replication parameters. We show that compared to the classical Tyler's estimator, demanding at least $n = p+1$ independent samples to guarantee the existence and uniqueness, the STyler requires a much less number of samples, with the exact value depending on the algebraic properties of the group $\mathcal{G}$ introduced below. We develop a simple fixed point iteration for computing the estimator and prove its convergence. In addition to the existence and uniqueness issues, an important criterion of evaluating the power of an estimator is its performance characteristics under specific population model assumptions. In this direction we provide high-probability error bounds demonstrating STyler's performance advantages in group symmetric elliptically contoured distributions. Namely, we show that the probability of large deviations of the inverse of the STyler is no longer governed by the high ambient dimension, as it is in the original Tyler's estimator, \cite{soloveychik2014non}, but rater by the lower intrinsic dimension. Below we deduce the exact dimensionality gain for general groups and calculate its value for the most common examples. Therefore, our derivation precisely quantifies the intuition that a smaller number of samples is required to achieve the same accuracy as in the unconstrained case.

In essence, we extend the contributions of \cite{shah2012group} on Gaussian group symmetry structures to a more complicated robust case. The extension is challenging due to three main reasons. First, the Gaussian MLE, namely the sample covariance, has a simple closed form, whereas Tyler's estimator is an iterative solution which is harder to analyze. Gaussian covariance estimation methods with and without regularization, are all based on the sample covariance which is a simple sufficient statistic. There is no such sufficient statistic in the non-Gaussian case, and the estimators must work with the samples themselves. Second, in Gaussian models, block sparsity in the covariance implies statistical independence, and this is not true in Angular models. This too complicates the analysis. Third, Gaussian MLE based optimizations are convex in the inverse covariance, and can therefore easily exploit symmetry constraints which are linear in the inverse covariance. Tyler's optimization is not convex, but $g$-convex and it is not clear whether it can exploit linear constraints. Our results show the symmetry constraints are also $g$-convex, and this is the underlying principle that allows us to exploit them efficiently. In particular, we demonstrate that being a solution to the constrained Tyler's optimization, the STyler is an MLE of a certain multivariate population, and therefore asymptotically reaches the Cramer-Rao lower Bound (CRB).

The rest of the text is organized as follows. First, we introduce necessary notations, definitions and auxiliary results. Then we demonstrate the power of the proposed method in the Gaussian setting and provide common examples of group symmetric structures in Section \ref{gr_sym_s}. In Section \ref{st_sec} we introduce the STyler estimator, formulate and prove our main results in Theorem \ref{eq_lem} and provide a brief discussion of the theorem. Section \ref{schur_sec} is devoted to the performance analysis of the STyler in elliptical populations and its main claim is Theorem \ref{thm:main_res}. Finally, we demonstrate the benefits of the suggested techniques by numerical simulations and provide a Conclusion section.

\subsection{Notation}
Denote by $\mathcal{S}(p)$ the linear space of $p \times p$ Hermitian matrices, by $\mathcal{P}(p) \subset \mathcal{S}(p)$ the cone of positive definite matrices. Other linear spaces and subspaces are denoted by capital letters, such as $V$. $\I_p$ or $\I$ denote the identity matrix of a proper dimension. For a matrix $\M$, we denote by $\M^T$ its transpose and by $\M^H$ its complex conjugate transpose. For a complex number $c$, $\bar c$ denots its conjugate and $j$ stands for the imaginary unit. Given a subset $\X$ of a linear space, $\sspan{\X}$ denotes the subspace it spans. Given a subspace $V$ of a unitary space, we denote by $\Pi_{V}$ the orthogonal projection operator onto $V$ or its matrix in a given basis. For a finite set $\mathcal{F},\; |\mathcal{F}|$ stands for its cardinality.

We endow $\mathcal{S}(p)$ with the scalar product $(\A,\B) = \Tr{\A\B}$, which induces Frobenius norm $\norm{\cdot}_F$ on it. $\norm{\cdot}$ will denote the Euclidean norm for vectors and $\norm{\cdot}_2$ - spectral norm for matrices. The linear space $\mathbb{C}^p$ is treated as a column vector space with the standard inner product. For a matrix $\A \in \mathcal{P}(p)$, denote by $\lambda_{\min}(\A)$ and $\lambda_{\max}(\A)$ its minimal and maximal eigenvalues, correspondingly. Given a square matrix $\A$, $|\A|$ stands for its determinant. $U(p)$ denotes the set of all $p \times p$ unitary matrices.

\section*{Acknowledgment}
The authors are grateful to Teng Zhang and the anonymous reviewers for their useful remarks and important suggestions, which helped us to significantly improve the paper.

\section{Group Symmetric Structure}
\label{gr_sym_s}
In this section we define group symmetric matrix sets, discuss their main properties and provide popular examples.

\begin{definition}
A set $\mathcal{G} \subset U(p)$ is referred to as a unitary matrix group, if $\I \in \mathcal{G}$, and for any $\U_1, \U_2 \in \mathcal{G},\; \U_1\cdot\U_2 \in \mathcal{G}$ together with $\U_1^{-1} \in \mathcal{G}$. Such a group is denoted by $\mathcal{G} \leqslant U(p)$.
\end{definition}
Given a set $\mathcal{F} \subset \mathbb{C}^p$, denote by $\mathcal{G}\mathcal{F} = \cup_{\K \in \mathcal{G}}\K\mathcal{F}$ its orbit under the group action. Below we consider only finite groups, i.e. groups with finite number of elements.

\begin{definition}
Let $\mathcal{G} \leqslant U(p)$ be a finite unitary group. Given a set $\mathcal{V} \subset \mathcal{S}(p)$ we denote by
\begin{equation}
\mathcal{V}^\mathcal{G} = \{\M \in \mathcal{V}\mid \K^H\M \K = \M, \forall \K \in \mathcal{G}\},
\label{inv_cond}
\end{equation}
its subset of matrices fixed by the conjugation action of the group $\mathcal{G}$.
\end{definition}

An equivalent definition is: $\M$ belongs to $\mathcal{V}^\mathcal{G}$ iff it commutes with all the elements of $\mathcal{G}$. Indeed, for any $\K\in U(p)$ and an arbitrary $\M\in\mathcal{S}(p)\;\colon\;\K^H \M \K=\M \Leftrightarrow \M \K= \K \M$.

Group symmetry can also be stated in terms of the inverse. For $\M\in\mathcal{P}(p)$ the condition (\ref{inv_cond}) is equivalent to $\K^H\M^{-1} \K = \M^{-1}, \; \forall \K \in \mathcal{G}$. Thus an invertible matrix is $\mathcal{G}$-invariant together with its inverse. Below we also make use of the following
\begin{definition}
\label{inv_ss}
We say that $L \subset \mathbb{C}^p$ is a $\mathcal G$-invariant subspace if for any $\x\in L$ and $\K \in \mathcal G$, $\K\x \in L$.
\end{definition}
Note that in Definition \ref{inv_ss} we require the subspace to coincide with its image under the group action, however we allow the images of individual vectors $\K\x$ not to be collinear with the vectors $\x$ themselves.

As explained in the introduction, the optimization machinery behind Tyler's estimator is the $g$-convexity of the associated negative log-likelihood function, \cite{rapcsak1991geodesic, wiesel2012geodesic, wiesel2012unified, wiesel2011regularized}. Just like classical convexity, this guarantees uniqueness and convergence properties of a $g$-convex optimization program over $g$-convex sets. Definitions and a brief review on $g$-convexity is available in the Appendix \ref{g_app}. For the purposes of our paper, the following result states the group symmetry constraints are also $g$-convex, and can be efficiently exploited in Tyler's estimator. 

\begin{theorem}
\label{g_th}
The set $\mathcal{P}(p)^\mathcal{G}$ is $g$-convex with respect to the Riemannian metric over the manifold $\mathcal{P}(p)$.
\begin{proof}
The proof can be found in Appendix \ref{g_app}.
\end{proof}
\end{theorem}

Given a finite unitary group $\mathcal{G}$, it is always possible to construct an orthonormal basis in which all the $\mathcal{G}$-invariant matrices have a certain block-diagonal structure depending only on $\mathcal{G}$. The formal statement reads as
\begin{theorem} (Classification Theorem, \cite{curtis1962representation})
\label{rep_basis}
Let $\mathcal{G}\leqslant U(p)$ be a finite unitary group acting on a unitary space $\mathbb{C}^p$, then there exists an orthonormal basis $\Q=\[\q_1,\dots,\q_p\]$ in $\mathbb{C}^p$ such that in this basis
any matrix $\A \in \mathcal{S}(p)^{\mathcal{G}}$ reads as
\begin{equation}
\A=
  \begin{bmatrix}
    \A_1 & 0  & \dots & 0\\
    0 & \A_2 & \dots & 0\\
    \vdots & \vdots & \ddots & \vdots\\
    0 & 0 & \dots & \A_m\\
  \end{bmatrix},
\label{shur_f}
\end{equation}
where the blocks $\A_i$ have the following structure:
\begin{equation}
\A_i=\I_{p_i} \otimes \B_i,\;\; i=1,\dots,m,
\end{equation}
with $\B_i \in \mathcal{S}(s_i)$. In particular, each block $\A_i$ is of size $p_is_i \times p_is_i$.
\end{theorem}
The values of $s_i$ and $p_i$ are completely determined by the group $\mathcal{G}$. The exact decomposition and ways to find it follow from the Wedderburn Structure Theorem and Maschke Theorem, \cite{feit1982representation, curtis1962representation}\footnote{Note that the Classification Theorem is only valid over algebraically closed fields, such as the field of complex numbers. This is the reason why we consider Tyler's complex version. A real valued extension of the Classification Theorem is discussed in \cite{maehara2010numerical}.}. Below, we list the specific values for structures which are common in modern covariance estimation.

The Classification Theorem plays a crucial role in group symmetry covariance estimation. It means that group symmetry implies sparsity and replication if the basis is appropriately chosen. This reduces the number of intrinsic degrees of freedom. In particular, let $\M \in \mathcal{S}(p)$ and partition it as
\begin{multline}
\M=
  \begin{bmatrix}
    \M_{11} & \M_{12}  & \dots & \M_{1m} \\
    \M_{21} & \M_{22} & \dots & \M_{2m}\\
    \vdots & \vdots & \ddots & \vdots\\
    \M_{m1} & \M_{m2} & \dots & \M_{mm}\\
  \end{bmatrix}, \\
\text{ where }
\M_{ii} =
  \begin{bmatrix}
    \M_{11}^i & \M_{12}^i  & \dots & \M_{1p_i}^i \\
    \M_{21}^i & \M_{22}^i & \dots & \M_{2p_i}^i\\
    \vdots & \vdots & \ddots & \vdots\\
    \M_{p_i1}^i & \M_{p_i2}^i & \dots & \M_{p_ip_i}^i\\
  \end{bmatrix}.
\end{multline}
Denote the orthogonal projection $\Pi_{\mathcal{S}(p)^\mathcal{G}} \colon \mathcal{S}(p) \to \mathcal{S}(p)^\mathcal{G}$, also referred to as the Reynolds operator \cite{reynolds1894dynamical}, by $\Pi_{\mathcal{G}}$. One can show that it is given by averaging over the group action,
\begin{equation}
\Pi_{\mathcal{G}}(\M) = \frac{1}{|\mathcal{G}|}\sum_{\K \in \mathcal{G}} \K\M\K^H.
\label{rop}
\end{equation}
Now, in the basis $\Q$ specified by Theorem \ref{g_th},
\begin{equation}
\Pi_{\mathcal{G}}(\M) =
  \begin{bmatrix}
    \Pi_{\mathcal{G}}(\M_{11}) & 0  & \dots & 0 \\
    0 & \Pi_{\mathcal{G}}(\M_{22}) & \dots & 0\\
    \vdots & \vdots & \ddots & \vdots\\
    0 & 0 & \dots & \Pi_{\mathcal{G}}(\M_{mm})\\
  \end{bmatrix},
\end{equation}
where 
\begin{equation}
\Pi_{\mathcal{G}}(\M_{ii}) = \I_{p_i} \otimes \frac{1}{p_i}\sum_{j=1}^{p_i} \M_{jj}^i,\;\; i=1,\dots,m.
\end{equation}
In other words, the projection replaces the off-diagonal blocks with zeros and the diagonal ones by their average. This suggests the definition of two quantities. First, the sparsity factor is the number of nonzero elements divided by the total number of elements:
\begin{equation}
\rho(\mathcal{G}) = \frac{\sum_{i=1}^m p_i s_i^2}{p^2}.
\end{equation}
Second, the degrees of freedom factor, which takes the averaging into account, and is the ratio between the intrinsic degrees of freedom and the ambient dimension:
\begin{equation}
\delta(\mathcal{G}) = \frac{ \max_{i=1}^m \frac{s_i}{p_i}}{p}.
\end{equation}
The main message of this paper is that estimators that exploit group symmetry enjoy these gains in their existence, uniqueness and sample complexity properties. 

The following lemma quantifies the gain in rank by applying the Reynolds operator to a rank one random matrix.
\begin{lemma}
\label{spanlem}
Let $\x\in \mathbb C^p$ be continuously distributed, and $\Pi_i$ - the $i$-th diagonal block of matrix $\Pi_{\sspan{\mathcal{G}\x}}$, as in (\ref{shur_f}), then a.s.
\begin{equation}
\label{comp_size}
\rank{\Pi_i} = p_i\min[s_i,p_i],\;\; \forall i = 1,\dots, m,
\end{equation}
and, therefore,
\begin{equation}
\dim\sspan{\mathcal{G}\x} = \rank{\Pi_{\sspan{\mathcal{G}\x}}} =\sum_{i=1}^m p_i\min[s_i,p_i].
\end{equation}
\end{lemma}
\begin{proof}
Denote $\M = \x\x^H$, then $\Pi_{\sspan{\mathcal{G}\x}}$ is the orthogonal projector onto the image of $\Pi_{\mathcal G}(\M)$. Now we infer that
\begin{equation}
\rank{\Pi_i} = \rank{\Pi_{\mathcal{G}}(\M_{ii})} = p_i \cdot \rank{\frac{1}{p_i}\sum_{j=1}^{p_i} \M_{jj}^i}.
\end{equation}
Note that $\rank{\sum_{j=1}^{p_i} \M_{jj}^i} = \min[s_i, p_i]$ with probability one for a continuously distributed $\x$, to get the desired.
\end{proof}

The power of Lemma \ref{spanlem} can be illustrated on the sample covariance matrix as detailed in \cite{shah2012group}. Given $n$ i.i.d. zero mean Gaussian vectors $\X = \{\x_1,\dots,\x_n\,|\,\x_i \in \mathbb{C}^p,\, i=1,\dots,n\}$, the most common covariance estimator is the sample covariance
\begin{equation}
\label{SCM}
\S =   \frac{1}{n}\sum_{i=1}^n \x_i\x_i^H.
\end{equation}
It is well known that $\S$ is a.s. of full rank when
\begin{equation}
 n\geqslant p.
\end{equation}
In group symmetric distributions, \cite{shah2012group}  proposed to improve this estimator using Reynold's averaging
\begin{equation}
\label{SCMRey}
\S^{\mathcal{G}} = \Pi_{\mathcal{G}}(\S) = \frac{1}{n|\mathcal{G}|}\sum_{i=1}^n \sum_{\K \in \mathcal{G}} \K\x_i\x_i^H\K^H.
\end{equation}

\begin{corollary}
\label{lemma:isotypicBound}
Let $\X\subseteq \mathbb C^p$ be a set of independent continuously distributed vectors, and $\Pi_i$ - the $i$-th block of matrix $\Pi_{\sspan{\mathcal{G}\X}}$, as in (\ref{shur_f}), then 
\begin{equation}
\rank{\Pi_i} = p_i\min[s_i,np_i],\;\; \forall i = 1,\dots, m,\;\;\text{a.s.},
\label{cor_f_e}
\end{equation}
\begin{equation}
\rank{\S^{\mathcal{G}}} = \sum_{i=1}^m p_i\min[s_i,n p_i],\;\;\text{a.s.}
\label{scm_g_d}
\end{equation}
\end{corollary}
\begin{proof}
The proof can be found in Appendix \ref{cor1_app}.
\end{proof}
Recall that $p=\sum_{i=1}^m p_is_i$ to obtain that $\S^{\mathcal{G}}$ is a.s. full rank when 
\begin{equation}
n \geqslant \delta(\mathcal{G}) p.
\label{scm_b}
\end{equation}
As expected, the required number of samples is reduced by the degrees of freedom factor.

We conclude this brief introduction to group symmetry by listing a few examples of such structures which are ubiquitous in statistical signal processing:
\begin{itemize}
[leftmargin=*]
\item {\bf{Multiples of identity}}:
The simplest case of the group symmetry is the class of matrices of the form $\C = c\I$, where $c$ is a complex scalar. Obviously, such matrices commute with the whole $\mathcal{G} = U(p)$ (which is not a finite group) and already possess diagonal form, thus no basis change is required. Here $m=1, p_1=p, s_1=1, \rho(\mathcal{G})=\delta(\mathcal{G}) = 1/p$.

\item {\bf{Matrices with equal variances and covariances}}:
The next family of group symmetric covariances is obtained by taking $\mathcal{G}$ to be all the permutations on the coordinates of $p$-dimensional vectors, the $S_p$ group. In such a case the only matrices belonging to $\mathcal{S}(p)^\mathcal{G}$ are $\I$ and $\e\e^H$, where $\e = [1,\dots,1]^H$, and their linear combinations
\begin{equation}
\C =
 \begin{pmatrix}
  a & b & \dots & b \\
  b & a & \dots & b\\
  \vdots & \vdots & \ddots & \vdots \\
  b & b & \dots & a
  \end{pmatrix}.
\end{equation}
The orthonormal eigenbasis is constructed simply by taking one of the vectors to be $\frac{1}{\sqrt{p}}\e$ and completing it to an orthonormal basis arbitrarily. In this basis $\C$ reads as
\begin{equation}
\Q_l^H\C\Q_l =
 \begin{pmatrix}
  a+b(p-1) & 0  & \dots & 0 \\
  0 & a-b  & \dots & 0\\
  \vdots & \vdots  & \ddots & \vdots \\
  0 & 0 & \dots & a-b
  \end{pmatrix}.
\end{equation}
In this example $m=2, p_1=1, s_1=1, p_2=p-1, s_2=1, \rho(\mathcal{G}) = \delta(\mathcal{G}) = 1/p$. This is in fact a particular case of a more general covariance model. Let the group $\mathcal{G}_1$ consisting of $k\times k$ permutation matrices act on the first $k$ coordinates of a vector. Consider the infinite group of $p \times p$ unitary matrices $\mathcal{G} = \mathcal{G}_1 \times U(p-k)$, its fixed point set $\mathcal{P}(p)^\mathcal{G}$ consists of matrices of the form
\begin{equation}
\C =
 \begin{pmatrix}
  a & b & \dots & b & 0 & \dots & 0 \\
  b & a & \dots & b & 0 & \dots & 0 \\
  \vdots & \vdots & \ddots & \vdots & \vdots & \ddots & \vdots\\
  b & b & \dots & a & 0 & \dots & 0\\
  0 & 0 & \dots & 0 & c & \dots & 0 \\
  \vdots & \vdots & \ddots & \vdots & \vdots & \ddots & \vdots \\
  0 & 0 & \dots & 0 & 0 & \dots & c \\
  \end{pmatrix}.
\end{equation}
This example demonstrates a general idea, that simple models can play the role of building blocks for more involved ones, we just directly multiply the groups acting on the direct summands of the underlaying space to obtain them.

\item {\bf{Circulant}}:
\label{circ_def}
Raising the complexity, the next common class of group symmetric covariances is the set of Hermitian circulant matrices defined as
\begin{equation}
\C =
 \begin{pmatrix}
  c_1 & c_2 & c_3 & \dots & c_{p} \\
  c_{p} & c_1 & c_2 & \dots & c_{p-1}\\
  c_{p-1} & c_{p} & c_1 & \dots & c_{p-2}\\
  \vdots & \vdots & \vdots & \ddots & \vdots \\
  c_2 & c_3 & c_4 & \dots & c_1
  \end{pmatrix},
\label{circ_struct}
\end{equation}
with the natural Hermitian conditions $c_2 = \bar c_p$, etc. Such matrices are typically used as approximations to Toeplitz matrices which are associated with signals that obey periodic stochastic properties. These processes are frequently encountered in engineering and arise in a variety of contexts such as communication systems, radar detection, and the study of vibrations in mechanical engineering, \cite{wiesel2013time,dembo1989embedding, cai2013optimal}.
It is easy to see that the set of circulant matrices is invariant under conjugation by the shift matrix
\begin{equation}
\bm\Pi =
 \begin{pmatrix}
  0 & 1 & 0 & \dots & 0 \\
  0 & 0 & 1 & \dots & 0 \\
  \vdots & \vdots & \vdots & \ddots & \vdots \\
  0 & 0 & 0 & \dots & 1 \\
  1 & 0 & 0 & \dots & 0 \\
  \end{pmatrix},
\label{shift_mat}
\end{equation}
and all its powers $\bm\Pi^i,i=1,\dots,p$, forming a cyclic group of order $p$. It is well known that there exists an orthonormal basis diagonalizing the circulant matrices, which is given by the FFT matrix
\begin{equation}
\Q_c =
 \frac{1}{\sqrt{p}}\begin{pmatrix}
  1 & 1 &  \dots & 1 \\
  w_0 & w_1 & \dots & w_{p-1} \\
  \vdots & \vdots &  \ddots & \vdots \\
  w_0^{p-1} & w_1^{p-1} & \dots & w_{p-1}^{p-1} \\
  \end{pmatrix},
\label{circ_rot}
\end{equation}
where $w_i = e^{2\pi ji/p}$ are the complex roots of unity. In this case $m=p,\; p_i=1,\; s_i=1,\; i=1,\dots m,\; \rho(\mathcal{G}) = \delta(\mathcal{G}) = 1/p$.

In all the examples considered up to now there exists a basis, in which all the elements of $\mathcal{S}(p)^{\mathcal{G}}$ are diagonal, thus $\rho(\mathcal{G}) = \delta(\mathcal{G}) = 1/p$. Remarkably, this implies that a single measurement is enough to get a full rank SCM a.s.

\item {\bf{Block-circulant}}:
\label{circ_def}
A natural generalization of the class of circulant matrices is the class of block-circulant matrices, which is a set of Hermitian $p \times p$ matrices with the structure
\begin{equation}
\C =
 \begin{pmatrix}
  \C_1 & \C_2 & \C_3 & \dots & \C_{p/d} \\
  \C_{p/d} & \C_1 & \C_2 & \dots & \C_{p/d-1}\\
  \C_{p/d-1} & \C_{p/d} & \C_{1} & \dots & \C_{p/d-2}\\
  \vdots & \vdots & \vdots & \ddots & \vdots \\
  \C_2 & \C_3 & \C_4 & \dots & \C_1
  \end{pmatrix},
\end{equation}
where $\C_i$ are $d \times d$ square blocks. The set of block-circulant matrices is invariant under $\bm\Pi^d$-conjugation and its powers $\bm\Pi^{kd},\;k=1,\dots,p/d$, forming a cyclic group of order $p/d$. In this case, similarly, there exists an orthonormal basis bringing the block-circulant matrices to the block-diagonal form, which reads as \cite{wang2004eigenvalues}
\begin{equation}
\Q_{bc} = \Q_c \otimes \I_{d}.
\end{equation}
Here $m=p/d,\; p_i=1,\; s_i=d,\; i=1,\dots m,\; \rho(\mathcal{G}) = \delta(\mathcal{G}) = d/p$.

\item {\bf{PerHermitian}}:
\label{per_def}
Another popular class of group symmetric covariances is the set of Hermitian PerHermitian matrices, i.e., matrices which are in addition Hermitian with respect to the northeast-to-southwest diagonal. This condition can be concisely written as
\begin{equation}
\C \J = \J \C,
\end{equation}
where $\J$ is the exchange matrix
\begin{equation}
\J =
 \begin{pmatrix}
  0 & 0 & \dots & 1 \\
  \vdots & \vdots&  \udots & \vdots \\
  0 & 1 & \dots & 0 \\
  1 & 0 & \dots & 0 \\
  \end{pmatrix}.
\end{equation}
PerHermitian matrices and their real analog - persymmetric matrices are commonly encountered in radar systems using a symmetrically spaced linear array with constant pulse repetition interval, \cite{pailloux2011persymmetric}. This structure information can be exploited to improve detection performance, \cite{de2003maximum, pailloux2011persymmetric}. Note that the $\J$ generates a group of two elements $\{\I,\J\}$. Any Hermitian perHermitian matrix $\P$ can be unitarily transformed to a block-diagonal matrix with blocks of size $p/2$ for even $p$ by the following basis change matrix
\begin{equation}
\Q_p =
 \frac{1}{\sqrt{2}}\begin{pmatrix}
  \I & \I\\
  \J & -\J\\
  \end{pmatrix}.
\end{equation}
It is easy to show that $m=2, p_i=1, s_i=p/2, i=1,2,\; \rho(\mathcal{G}) = \delta(\mathcal{G}) = 1/2$.
\item {\bf{Proper Quaternion Covariance Matrix}}:
Many physical processes can be conveniently described in terms of quaternion signals. Quaternion numbers are a generalization of complex numbers to numbers with $4$ real elements, so that a length $p$ quaternion vector can be represented as $4p$ real or a $2p$ complex vector. Typical applications are complex electromagnetic signals with two polarizations \cite{sloin2013gaussian,miron2006quaternion}. Here we use the complex representation of quaternions. It is common to consider proper distributions, which are invariant to certain classes of quaternion rotations, \cite{sloin2013gaussian, ginzberg2011testing}. Among the different kinds of properness we choose the following: given a centered quaternion random vector
\begin{equation*}
\q = \x_1+k\x_2,\;\; \x_i \in \mathbb{C}^p,\;\; k\neq j,\; k^2=-1,
\end{equation*}
we say it is proper if $\x_1, \x_2$ are both independent complex proper with the same covariances and $\q$ is $k$-proper meaning that
\begin{equation}
\mathbb{E}[\q\q^H] = 0.
\end{equation}
This definition implies that the distribution of $\q$ is completely defined by the matrix
\begin{equation}
\C_H = \mathbb{E}[\q\q^{kH}],
\end{equation}
where
\begin{equation}
\q^{kH} = \x_1^H-k\x_2^H
\end{equation}
is the conjugation transpose with respect to $k$ and the subscript $H$ stands for Hamiltion, the discoverer of quaternions. The $2p\times 2p$ complex Hermitian form of $\C_H$ reads as
\begin{equation}
\C_C = \begin{pmatrix}
  \C_{11} & -\C_{12} \\
  \C_{12} & \C_{11}
  \end{pmatrix},
\end{equation}
where $\C_{11} = \mathbb{E}[\x_1\x_1^H],\; \C_{12} = \mathbb{E}[\x_1\x_2^H]$. Proper quaternion covariances and only them are invariant under the conjugation by the matrix
\begin{equation}
\Y =
\begin{pmatrix}
  0 & -1 \\
  1 & 0
  \end{pmatrix} \otimes \I_p,
\label{prop_q_g}
\end{equation}
forming a cyclic group of second order. The corresponding orthonormal basis turning $\C_H$ into the block-diagonal form with two $p\times p$ blocks reads as
\begin{equation}
\Q_H = \frac{1}{\sqrt{2}}\begin{pmatrix}
  1 & -1 \\
  j & j
  \end{pmatrix} \otimes \I_p.
\end{equation}
Here $m=2,\; p_i=1,\; s_i=2p/2=p,\; i=1,2,\;\rho(\mathcal{G}) = \delta(\mathcal{G}) = 1/2$.
\end{itemize}
As we have already mentioned and shown by examples, more involved symmetry groups can be constructed by using the simple ones as building blocks and superposing them via direct product.

\section{The STyler Estimator}
\label{st_sec}
In this section we introduce the STyler - a group symmetric version of Tyler's robust covariance estimator. Following the original derivation in \cite{tyler1987distribution}, we begin with an implicit definition, and then discuss its existence, uniqueness and convergence properties. Throughout this section, we do not assume any specific probabilistic model. 

Assume the data consists of $n$ complex vectors
\begin{equation}
\X = \{\x_1,\dots,\x_n \mid \x_i \in \mathbb{C}^p,\; i=1,\dots,n\}.
\end{equation}
In the Gaussian case, a natural covariance estimator using this data is the sample covariance in (\ref{SCM}). Its symmetric version is obtained by applying this formula to the synthetically rotated and replicated data ${\mathcal{G}}\X$. This leads to the Reynold's projection in (\ref{SCMRey}). Similarly, we define the STyler by applying Tyler's definition in (\ref{tylerequ}) to ${\mathcal{G}}\X$:
\begin{definition}
\label{def_s}
Any matrix satisfying
\begin{equation}
\widehat{\bm\Theta}^\mathcal{G} = \frac{p}{n|\mathcal{G}|}\sum_{i=1}^n \sum_{\K \in \mathcal{G}} \frac{\K\x_i\x_i^H\K^H}{\x_i^H\K^H\[\widehat{\bm\Theta}^\mathcal{G}\]^{-1}\K\x_i}
\label{def_ste_fp}
\end{equation}
is referred to as the STyler estimator.
\end{definition}

The following theorem characterizes the appealing properties of the STyler estimator.
\begin{theorem}
\label{eq_lem}
When $\X$ is sampled from a continuous distribution with independent samples and $n>\delta(\mathcal{G})p$, the STyler exists, belongs to $\mathcal{P}(p)^\mathcal{G}$ and is unique a.s., up to a scaling factor. It can be computed via the normalized fixed point iteration
\begin{equation}
\begin{cases}
\bm\Psi_{j+1}  = \sum_{i=1}^n \sum_{\K \in \mathcal{G}} \frac{\K\x_i\x_i^H\K^H}{\x_i^H\K^H\[\widehat{\bm\Theta}_j^\mathcal{G}\]^{-1}\K\x_i},\\
\widehat{\bm\Theta}_{j+1}^\mathcal{G} = \frac{\bm\Psi_{j+1}}{\Tr{\bm\Psi_{j+1}}}.
\label{def_ste_fpi}
\end{cases}
\end{equation}
which converges starting from any point in $\mathcal{P}(p)$.
\end{theorem}
\begin{proof}
Most of the properties are directly inherited from the original Tyler's estimator, with the advantage of requiring less samples due to their synthetic replication. In particular, we follow the derivations in \cite{zhang2012robust}. We begin by noting that the STyler can be interpreted as the minimizer of 
the $g$-convex function
\begin{equation}
F(\bm\Theta) = \frac{p}{n}\sum_{i=1}^n\log(\x_i^H\bm\Theta^{-1}\x_i)+\log|\bm\Theta|,
\label{target_gen}
\end{equation}
over the $g$-convex group symmetric set $\mathcal{P}(p)^\mathcal{G}$ (see Theorem \ref{g_th}). Indeed, when $\bm\Theta\in \mathcal{P}(p)^\mathcal{G}$ we have
\begin{equation}
F^{\mathcal G}(\bm\Theta) = \frac{p}{n|\mathcal{G}|}\sum_{i=1}^n\sum_{\K \in \mathcal{G}}\log(\x_i^H\K^H\bm\Theta^{-1}\K\x_i)+\log|\bm\Theta|.
\label{target_restr}
\end{equation}
Ignoring the constraints, the minimizer of this objective is Tyler's estimator with the replicated samples $\mathcal{G}\X$. It is easy to check that it belongs to $\mathcal{P}(p)^\mathcal{G}$ and is therefore also the minimizer subject to the constraints. 

{\em{Uniqueness:}} The uniqueness (up to scaling) follows from Lemma III.2 from \cite{zhang2012robust} applied to the target (\ref{target_restr}). Note that according to Corollary \ref{lemma:isotypicBound}, when $n>\delta(\mathcal{G})p$, the set $\mathcal G\X$ almost surely spans the space.

{\em{Existence:}} The solution to (\ref{target_restr}) is scale invariant, thus we fix the scale by restricting our attention to the set
\begin{equation}
\mathcal{M}^{\mathcal{G}} = \{\M\mid\M\in\mathcal{P}(p)^{\mathcal{G}}, \Tr{\M}=1\}.
\end{equation}
Note that in our case the set $\mathcal{G}\X$ may be dependent, both statistically and linearly, therefore, to prove that a minimizer exists we extend Lemma III.3 from \cite{zhang2012robust} to the group symmetric case in the following
\begin{lemma}
\label{geom_lem}
If for any random proper $\mathcal{G}$-invariant subspace $L \subset \mathbb{C}^p$
\begin{equation}
\frac{|\X \cap L|}{n} < \frac{\dim{L}}{p},
\label{gsc}
\end{equation}
then
\begin{equation}
F(\bm\Theta) \to + \infty, \text{ when } \mathcal{M}^\mathcal{G} \ni \bm\Theta \to \partial\mathcal{M}^\mathcal{G}.
\label{f_inf}
\end{equation}
\end{lemma}
\begin{proof}
The proof can be found in Appendix \ref{lem2_app}.
\end{proof}
Note that Lemma \ref{geom_lem} holds true for any random proper $\mathcal{G}$-invariant subspace $L$, and not only for an arbitrary fixed subspace. This technical detail is unavoidable due to the fact that we must allow $L$ to statistically depend on $\X$, which is a random set. As one can easily observe, if $L$ is a fixed nonzero subspace, the condition (\ref{f_inf}) verifies vacuously, since $\X \cap L$ is a.s. an empty set.

\begin{lemma}\label{lemma:asInequality}
If $\X \subset \mathbb{C}^p$ contains $n$ independent samples and
\begin{equation}
n > \max_{1\leqslant i \leqslant m}\(\frac{s_i}{p_i}\) = \delta(\mathcal G)p,
\end{equation}
then for any random proper $\mathcal G$-invariant subspace $L \subset \mathbb C^p$,
(\ref{gsc}) holds true.
\end{lemma}
\begin{proof}
The proof can be found in Appendix \ref{lem3_app}.
\end{proof}

Now Theorem \ref{g_th} together with Lemma \ref{geom_lem} implies the existence of the minimum of (\ref{target_gen}) restricted to $\mathcal M^{\mathcal G}$, which is also the minimum of (\ref{target_restr}).

{\em{Convergence:}} 
When we compute the STyler iteratively, using the scheme (\ref{def_ste_fpi}), we normalize the current approximation $\widehat{\bm\Theta}_i^\mathcal{G}$ by its trace on each step (see \cite{tyler1987distribution, zhang2012robust} for details). When $n>\delta(\mathcal{G})p,\; \mathcal{G}\X$ spans $\mathbb{C}^p$ a.s., hence the convergence of the iterative scheme from any starting point follows from Theorem III.4 from \cite{zhang2012robust}. 
\end{proof}

A few remarks are in order here. Note that unlike Theorem III.1 from \cite{zhang2012robust}, our Theorem \ref{eq_lem} does not hold for any data $\X$ in general position, but rather with probability one under any continuous distribution. This probabilistic condition cannot be discarded due to Lemma \ref{spanlem}, which may otherwise be wrong. As an example let us consider the circulant symmetry setting and take a vector $\x=[1,\dots,1]^T$, which is an eigenvector of all the shift matrices (\ref{shift_mat}). If such a vector belongs to the data set $\X$, its copies will all coincide with itself and will not contribute new information, which is formally reflected by the fact that Lemma \ref{spanlem} will fail. On the contrary, when the data is sampled randomly, Theorem \ref{eq_lem} suggests a very surprising and promising result on the number of required samples, compared to the classical Tyler's setting. In the same circulant case, according to Section \ref{circ_def}, only two samples are enough to guarantee the existence and uniqueness of the STyler no matter what the ambient dimension $p$ is. Similarly, significant benefits in the number of demanded measurements can be achieved under the other group symmetry constraints.

\section{STyler Performance Analysis in Elliptical Populations}
\label{schur_sec}
Having established the existence and uniqueness conditions for the STyler, in this section we proceed to a different important criterion of its power, namely, its performance properties. For this purpose we need to consider a specific class of populations. As we have already mentioned above, Tyler's estimator is closely related to a certain family of spherical distributions, \cite{greco2013cramer, soloveychik2014non}, and is actually an MLE of their shape matrix parameter.
In this section we briefly introduce this family of distributions, explain their relation to elliptical populations and claim high probability error guarantees on the STyler estimator.
\begin{definition}
\label{def}
Assume $\bm\Theta_0 \in \mathcal{P}(p)$, then the function
\begin{equation}
p(\x) = \frac{(p-1)!}{\pi^p} \frac{1}{|\bm\Theta_0|(\x^H  \bm\Theta_0^{-1}\x)^p}
\label{tyler_distr}
\end{equation}
is a probability density function of a vector $\x \in \mathbb{C}^p$ lying on a unit sphere. This distribution is usually referred to as the Complex Angular Elliptical (CAE) distribution, \cite{greco2013cramer}, and we denote it as $\x \sim \mathcal{U}(\bm\Theta_0)$. The matrix $\bm\Theta_0$ is referred to as a shape matrix of the distribution and is a multiple of the covariance matrix of $\x$.
\end{definition}

The CAE distribution is a straight forward extension of its real prototype, the real angular central Gaussian distribution, \cite{tyler1987statistical}. CAE is closely related to the class of complex GE distributions, which includes Gaussian, compound Gaussian, elliptical, skew-elliptical, CAE and other distributions, \cite{frahm2007tyler}. An important property of the GE family is that the shape matrix of a population does not change when the vector is divided by its Euclidean norm \cite{frahm2004generalized, frahm2007tyler}. After normalization, any GE vector becomes CAE distributed. This allows us to treat all these distributions together using Tyler's estimator, which is the MLE of the shape matrix parameter in CAE populations and is unbiased, \cite{greco2013cramer, soloveychik2014non}. Being an MLE, Tyler's estimator is known to be asymptotically statistically efficient for CAE populations, \cite{tyler1987distribution,pascal2008covariance}, and to reach the Cramer-Rao lower Bound (CRB), \cite{greco2013cramer}.

Since we focus on the group symmetric scenario, we would like to derive an MLE of the CAE shape matrix under the prior $\bm\Theta_0 \in \mathcal{P}(p)^\mathcal{G}$. As we have already mentioned in the proof of Theorem \ref{eq_lem}, the target (\ref{target_gen}), which is a negative log-likelihood of the CAE population (\ref{tyler_distr}), is a $g$-convex function. Together with the $g$-convexity of the group symmetric constraints, Theorem \ref{g_th}, this ensures that the constrained MLE can be efficiently found. Theorem \ref{eq_lem} also suggests that this group symmetric MLE is given by the STyler estimator (\ref{def_ste_fp}). As a corollary, this implies that when the true covariance is group symmetric, the STyler is consistent and asymptotically statistically efficient.

Next we claim a high probability error bound for the STyler by using the Classification Theorem. Indeed, we know that when the true covariance matrix $\bm\Theta_0$ is $\mathcal{G}$-invariant, there exists an orthonormal basis, in which $\bm\Theta_0$ posses sparse block-diagonal form. Therefore, the actual number of parameters to be estimated (the dimension of the subspace $\mathcal{S}(p)^\mathcal{G}$) is much smaller that the dimension of the ambient space $\mathcal{S}(p)$. This reduction in the number of active degrees of freedom allows for significant improvement of the error bounds, demonstrated below.

Let $\bm\Omega_0=\bm\Theta_0^{-1}, \;\underline \lambda=\lambda_{\min}(\bm\Theta_0) = \lambda_{\max}^{-1}(\bm\Omega_0) > 0$ and set
\begin{equation}
\cos{\phi_0} = \frac{\Tr{\bm\Omega_0}}{\sqrt{p}\norm{\bm\Omega_0}_F} > 0.
\label{ang_def}
\end{equation}

\begin{theorem}
\label{thm:main_res}
Given $n > \delta(\mathcal{G})p$ i.i.d. copies of $\x \sim \mathcal{U}(\bm\Theta_0)$, for $\theta\geqslant 0$ with probability at least
\begin{align}
&1-2\exp\(\frac{-\theta^2}{2(1+1.7\frac{\theta}{\sqrt{\rho(\mathcal{G})n}})}\) \\
&-2p^2\exp\(-\frac{n\cos^2{\phi_0}}{80\ln(7p)(1+\frac{1}{p})}\)\(1+\frac{8\cdot10^3(1+\frac{1}{p})^4}{n\cos^8{\phi_0}}\), \nonumber
\end{align}
the STyler estimator scaled by the condition $\Tr{[\widehat{\bm\Theta}^\mathcal{G}]^{-1}} = \Tr{\bm\Theta_0^{-1}}$ satisfies
\begin{equation}
\norm{[\widehat{\bm\Theta}^\mathcal{G}]^{-1} - \bm\Theta_0^{-1}}_F \leqslant \sqrt{\rho(\mathcal{G})}\frac{10\theta}{\underline\lambda \cos^2{\phi_0}}\frac{p+1}{\sqrt{n}}.
\label{fin_bound}
\end{equation}
\end{theorem}
\begin{proof}
The proof is based on the proof of Tyler's estimator error bound in \cite{soloveychik2014non}. The minor technical changes are due to transition to the complex case, which is straight forward, and only affects constants. The only significant core change is due to application of Theorem \ref{rep_basis}, which consists in using Corollary \ref{dim_r_cor} from Appendix \ref{perf_app} instead of Lemma 2 in \cite{soloveychik2014non}. This allows us to obtain a $\sqrt{\rho(\mathcal{G})}$ factor improvement in the bound (\ref{fin_bound}).
\end{proof}

This theorem basically claims that unlike the original Tyler's estimator, whose inverse's high-probability Frobenius norm error depends on $p$ and $n$ as $\frac{p}{\sqrt{n}}$, \cite{soloveychik2014non}, STyler's error is reduced by a factor of $\sqrt{\rho(\mathcal{G})}$. Moreover, this bound is guaranteed to be reliable already for $n > \delta(\mathcal{G})p$ samples, and does not require $p+1$ measurements as in Tyler's case. All this shows that in the STyler, the same level of accuracy can be achieved with a significantly reduces number of samples.

\section{Numerical Simulations}
\begin{figure}[!t]
\includegraphics[width=3.6in]{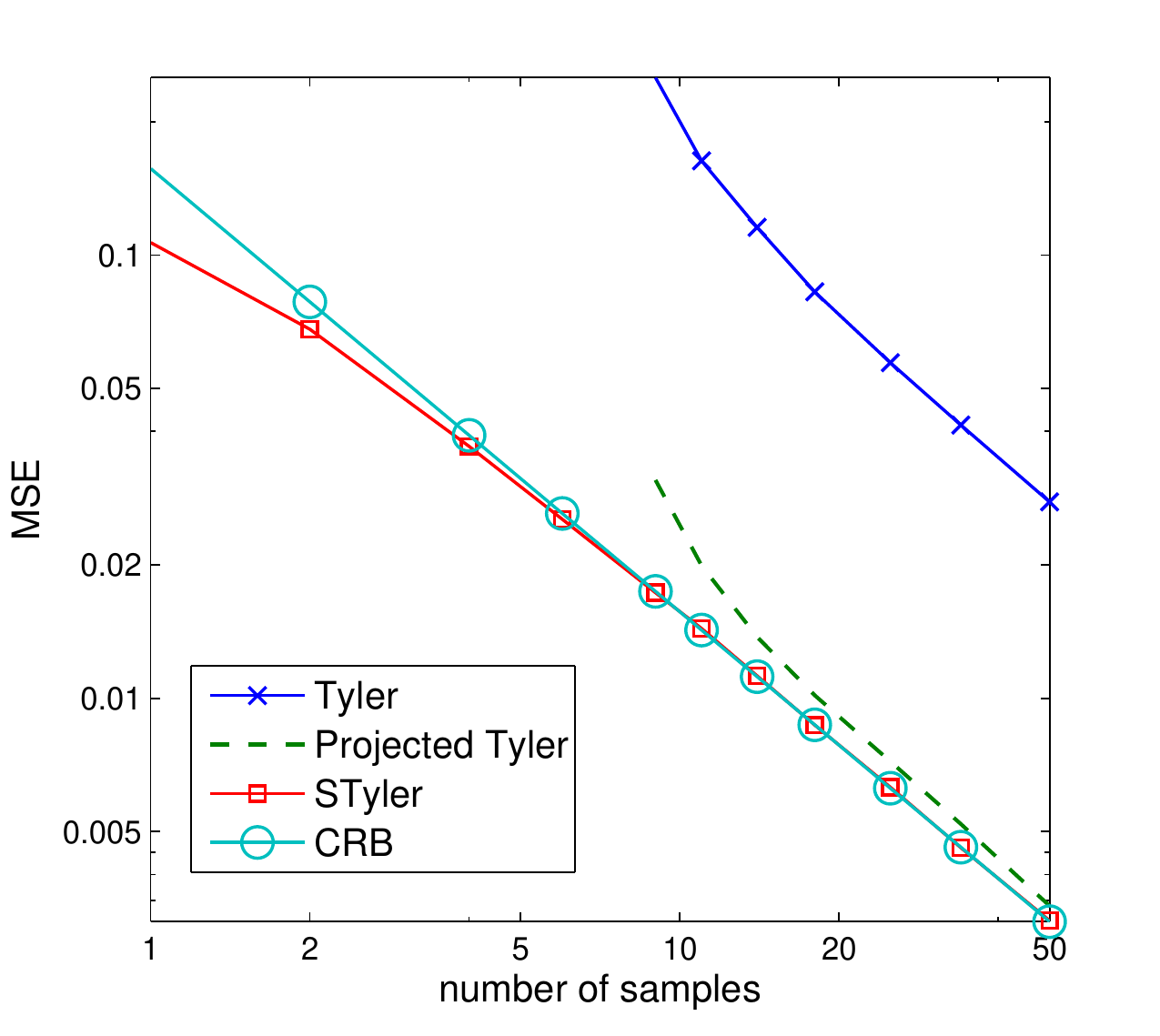}
\caption{STyler's performance in circulant case, $p=8$.}
\label{pic_c}
\includegraphics[width=3.6in]{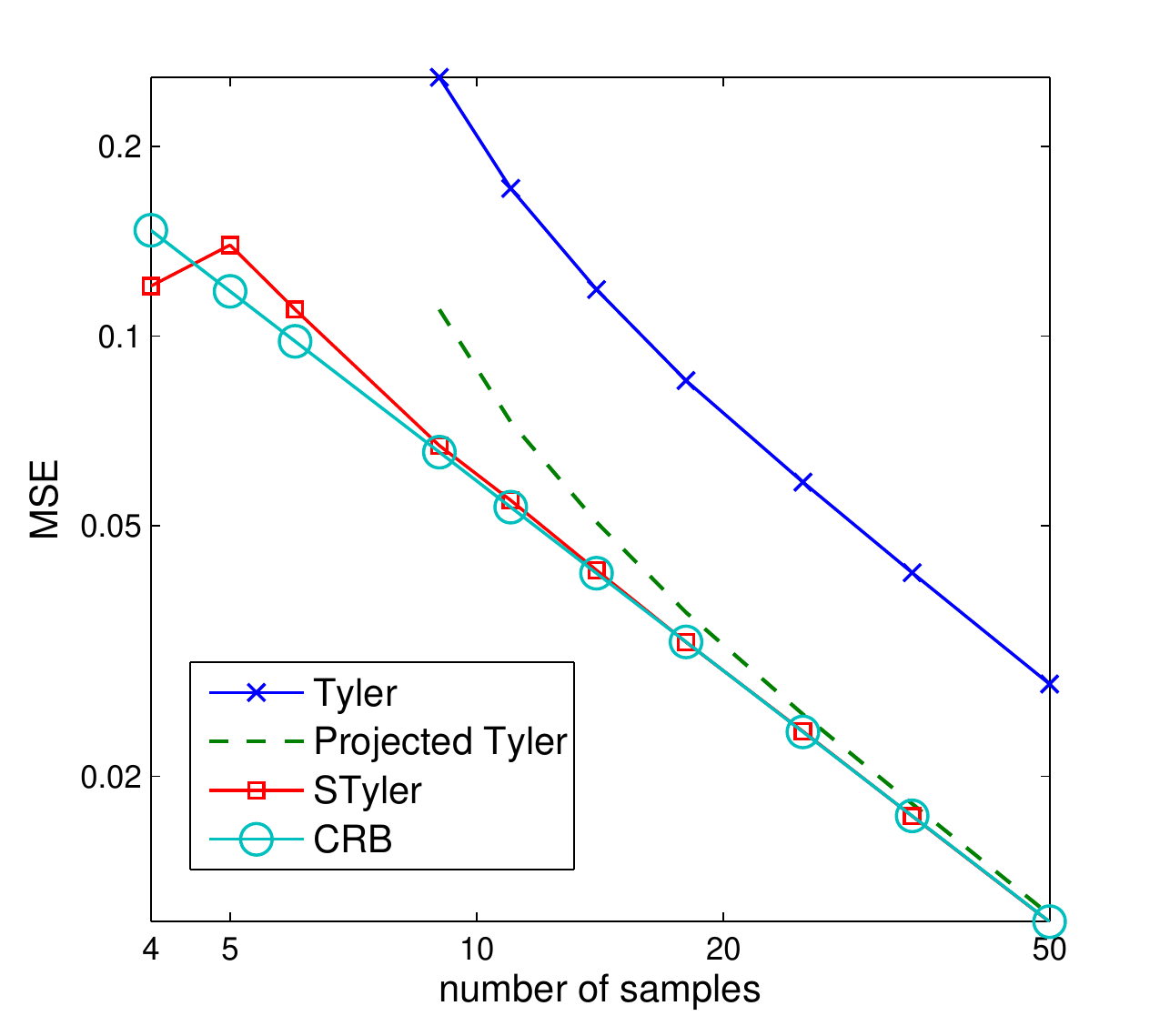}
\caption{STyler's performance in proper quaternion case, $p=8$.}
\label{pic_p}
\end{figure}
\label{appl}
In this section we present numerical simulations demonstrating the performance advantages of the STyler compared to Tyler's estimator and Tyler's projection onto the known structure set. The plots show the Mean Squared Errors (MSE) of the estimators defined as
\begin{equation}
\text{MSE}(\widehat{\bm\Theta}) = \mathbb{E}[\widehat{\bm\Theta} - \bm\Theta_0]^2,
\end{equation}
and the corresponding Cramer-Rao Bounds (CRBs) as functions of the number of samples. The CRBs bound from below variances of any unbiased estimators and in our case are obtained using the formulae from section III of \cite{soloveychik2014tyler}, by plugging the bases of the subspaces $\mathcal{S}(p)^{\mathcal{G}}$. Figure \ref{pic_c} shows both the performance advantages of the STyler in a $8$ dimensional circulant population, and the fact that $\delta(\mathcal{G})+1 = 2$ samples are enough for the STyler to exist and be unique. The true covariance matrix was a random circulant matrix with positive spectrum. Recall that circulant matrices are $\mathcal{G}$-invariant with $\mathcal{G}$ being the cyclic group of order $p$ generated by the shift matrix $\bm\Pi$ form (\ref{shift_mat}) of proper dimension.
Similarly, Figure \ref{pic_p} addresses the $8$ dimensional proper quaternion setting with the true covariance being a randomly generated proper positive definite matrix. In this case the corresponding group symmetry is induced by a cyclic group of order two, generated by a $8 \times 8$ matrix $\Y$ from (\ref{prop_q_g}). In this case $\delta(\mathcal{G})+1 = 5$ samples are sufficient to guarantee the existence and uniqueness of the constrained estimator.
Both graphs clearly confirm the performance benefits of the STyler estimator and the predicted by Theorem \ref{eq_lem} reduction in the demanded number of samples. In addition, the figures demonstrate that as the number of samples $n$ grows large, the log-scale performance gap becomes constant, since both Tyler's estimator and the STyler approach the corresponding CRB lines. We also note that the computational complexity of a single STyler's iteration is only at most $|\mathcal{G}|$ times larger than that of Tyler's estimator due to the increased number of summands in (\ref{def_ste_fp}) compared to (\ref{tylerequ}).

\section{Conclusion}
In the recent years robust covariance matrix estimation has become the cornerstone of many engineering applications. One of the most powerful and popular approaches to this task is to use the so called M-estimators, and in particular, Tyler's fixed point estimator. When the number of samples is not large enough to guarantee good estimation precision, prior knowledge in form of structural constraints is usually introduced. In this paper we focus on the group symmetric matrix constraints. We develop a novel group symmetric analog of Tyler's covariance estimator (the STyler) and show that its existence and uniqueness are guaranteed under much weaker requirements on the number of sample measurements. Surprisingly, this STyler estimator is given by a fixed point equation, analogous to the one corresponding to the original Tyler's estimator. In addition, we derive high probability error bounds of the STyler, which improve upon the Tyler's estimator's ones quite significantly. The results are supported by numerical simulations.

\appendices
\section{}
\label{g_app}
We briefly mention a few notions from the theory of smooth Riemannian manifolds. See \cite{rapcsak1991geodesic, wiesel2012geodesic} and references therein for a more detailed exposition.
\begin{definition}
With each pair of matrices $\M_0,\M_1\in
\mathcal{P}(p)$ we associate the geodesic curve
\begin{eqnarray}\label{geodesics}
 \M_t=\M_0^{\frac{1}{2}}\(\M_0^{-\frac{1}{2}}\M_1\M_0^{-\frac{1}{2}}\)^t\M_0^{\frac{1}{2}},\quad
 t\in[0,1].
\end{eqnarray}
\end{definition}
Note that $\M_t \in \mathcal{P}(p), t\in[0,1]$. Geodesic curves play a role on the smooth Riemannian manifolds similar to that of line segments in Euclidean spaces. Defining geodesic curves on a manifold is closely related to endowing the manifold with specific metric (or inner product) and taking these curves as the shortest paths between points of $\mathcal{P}(p)$ in this metric.
\begin{definition}
 A set ${\mathcal{N}} \subset \mathcal{P}(p)$ is called $g$-convex if for any $\M_0,\M_1\in
 \mathcal{N}$ the geodesic $\M_t$ lies in $\mathcal{N}$.
\end{definition}
\begin{definition}
Given a $g$-convex subset $\mathcal{N} \subset \mathcal{P}(p)$,
we say that a function $f$ is $g$-convex on $\mathcal{N}$ if for any two matrices
$\M_0,\M_1 \in \mathcal{N}, f(\M_t)\leqslant tf(\M_0)+(1-t)f(\M_1), \forall t \in [0,1]$.
\end{definition}

\begin{proof}[Proof of Theorem \ref{g_th}]
As we have already noted $\M = \K^H \M \K$ is equivalent to $\M$ and $\K$ being commutative. Now, assume $\M_0, \M_1 \in \mathcal{P}(p)^\mathcal{G}$. Let us show that the geodesic (\ref{geodesics}) lies in $\mathcal{P}(p)^\mathcal{G}$. Choose $\K \in \mathcal{G}$, $\M_0 \K = \K \M_0, \M_1 \K = \K \M_1$. Let $\M \in \mathcal{S}(p)$ then it is diagonalizable and let $\M=\Q^H\bm\Lambda\Q$ be its spectral decomposition, where $\bm\Lambda$ is diagonal and $\Q \in \mathcal{U}(p)$. Given a smooth function $f\colon\mathbb{C} \to \mathbb{C}$, we extend its action to $\mathcal{S}(p)$ in the functional way as $f(\M) = \Q^Hf(\bm\Lambda)\Q$, where $f$ acts on the diagonal entries of $\bm\Lambda$ elementwise.
$\M \in \mathcal{S}(p)$ commutes with $\P$ iff $f(\M)$ commutes with $\P$ for any smooth function $f$, also if two matrices $\A$ and $\B$ commute with $\P$, then their product $\A\B$ commutes with $\P$.
This implies that $\M_0^{-\frac{1}{2}}\M_1\M_0^{-\frac{1}{2}}$ commutes with $\K$,
thus $\(\M_0^{-\frac{1}{2}}\M_1\M_0^{-\frac{1}{2}}\)^t$ also commutes with $\K$ and the whole $\M_t$ commutes with $\K$. Therefore, the geodesic (\ref{geodesics}) lies in $\mathcal{P}(p)^\mathcal{G}$ and the set $\mathcal{P}(p)^\mathcal{G}$ is $g$-convex.
\end{proof}

\section{}
\label{cor1_app}
\begin{proof}[Proof of Corollary \ref{lemma:isotypicBound}]
The proof is by induction on $n$. Formula (\ref{comp_size}) from Lemma \ref{spanlem} provides the case $n=1$. Let $\x\in \X$, $L = \sspan{\mathcal{G}\x}$, $K = L^\perp$, and $\Y$ - be the projection of $\X\setminus\{\x\}$ onto $K$. Since $\X$ is independent, we may assume that $K$ is non-random and $\Y$ is continuously distributed on $K$. Then, $K = \mathbb C^{p'}$ with $p' = \sum_{i=1}^m s'_i p_i$, and $s'_i = s_i - \min[s_i,p_i]$ (if $s'_i = 0$ we drop the corresponding summand). Now (\ref{cor_f_e}) follows by induction applied to $\Y$ in $\mathbb C^{p'}$. And (\ref{scm_g_d}) follows immediately.
\end{proof}

\section{}
\label{lem2_app}
To clearly separate the main idea of the proof of Lemma \ref{geom_lem} from the technical details given in the auxiliary lemmas below, we provide a brief sketch of the proof emphasizing the most important steps. Consider a converging sequence $\{\bm\Theta_k\}_k \subset \mathcal{M}^\mathcal{G}$ and assume we can find an orthonormal basis $\Y = \{\y_1,\dots,\y_p\}$ in which $\bm\Theta_k$ read as 
\begin{equation}
\Y^H\bm\Theta_k\Y = 
\begin{pmatrix}
	{\lambda \I_d}&{0}\\
	{0}&{\I_{p-d}}\\
\end{pmatrix},\;\;\text{with}\;\; 
\lambda \to 0,\;\;\text{as}\;\; k\to \infty.
\end{equation}
Let $Y = \langle \y_1,\ldots \y_d\rangle$ and $K = Y^\perp$. We want to compute the main asymptotic term of $F(\bm\Theta_k)$, as $k\to \infty$ and to show it tends to positive infinity. For each $\x\in \X$, we have $\x^H\bm\Theta_k^{-1}\x \asymp \lambda^{-1}|\Pi_{Y}(\x)|^2$, where we write $a(\bm\Theta_k) \asymp b(\bm\Theta_k)$, if $\frac{a(\bm\Theta_k)}{b(\bm\Theta_k)} \to 1$ when $k \to \infty$. This implies that only $\x\notin K$ contribute to the asymptotic of $F(\bm\Theta_k)$. Namely,
\begin{equation}
\sum_{\x\in \X\setminus K} \log\;{\x^H\bm\Theta_k^{-1}\x} \asymp |\X\setminus K|\log\, \lambda^{-1} = \(n - |\X\cap K|\)\log\,\lambda^{-1},
\end{equation}
\begin{equation}
\frac{1}{p}\log\;|\bm\Theta_k| \asymp -\frac{d}{p} \log\,\lambda^{-1}= -\(1 - \frac{\dim{K}}{p}\)\log\,\lambda^{-1},
\end{equation}
thus,
\begin{equation}
F(\bm\Theta_k) \asymp \(\frac{\dim K}{p} - \frac{|\X\cap K|}{|\X|}\) \log\, \lambda^{-1}, \;\; \lambda\to 0+.
\end{equation}
If the expression in brackets is positive, $F(\bm\Theta_k)$ tends to $+\infty$ as $k\to \infty$. The requirement in Lemma \ref{geom_lem} is a variation of this condition adapted to the group symmetric setting. There is, however, a number of technical challenges, which we treat in the following auxiliary propositions. One of them is that we cannot in general find a basis in which all $\bm\Theta_k$ are diagonalizable and we do this only for the limiting point $\widetilde{\bm\Theta}$ of the sequence $\{\bm\Theta_k\}_k$. Second, following the remark after Lemmas \ref{geom_lem} and \ref{lemma:asInequality}, since $F(\bm\Theta_k)$ is a random function, the sequence $\{\bm\Theta_k\}_k$ must also be random, resulting in that the subspace $L$ appearing in Lemmas \ref{lemma:intersectBound} is random. In addition, when $\bm\Theta_k$ approach the boundary $\mathcal{M}^\mathcal{G}$, their eigenvalues may tend to zero with different rates, which may complicate the treatment. We show that these eigenvalues approaching zero split into groups in such a way that in each group having fixed rate of convergence, the corresponding eigenvectors form a $\mathcal{G}$-invariant subspace. These subspaces form flags appearing in Lemmas \ref{lemma:SAdditiv} and \ref{lemma:PartitCoeff}. Finally, we determine those eigenvalues that contribute to the main asymptotic term of $F(\bm\Theta_k)$ and calculate their number.

\begin{definition}
Let $\X\subseteq \mathbb{C}^p$ be a finite subset, $\mathcal F = \{\mathbb{C}^p=V_1\supsetneq V_2\supsetneq \ldots V_s\supsetneq V_{s+1} \supseteq 0\}$ be a flag (a sequence of proper subspaces) of length $s$ on $\mathbb{C}^p$, define
\begin{equation}
\Delta(\mathcal F, \X)_{i\,j} = \dim V_{i} - \dim V_j -\frac{\dim V_1}{|\X|}\left(|\X\cap V_i| - |\X\cap V_j|\right),
\end{equation}
where $1\leqslant i,\,j\leqslant s+1$. In addition, given a decreasing sequence
\begin{equation}
\r =\{r_1>\ldots>r_s\} \subset \mathbb{R}
\label{seq_r}
\end{equation}
of length $s$, define
\begin{equation}
S(\mathcal F, \X, \r) = \sum_{i=1}^s r_i \Delta(\mathcal F,\X)_{i\,i+1}.
\end{equation}
\end{definition}

\begin{lemma}\label{lemma:SAdditiv}
Let $\X\subseteq \mathbb{C}^p$ be a finite subset and $\mathcal F$ be a flag of length $s$ on $\mathbb{C}^p$, then
\begin{equation}
\Delta(\mathcal F, \X)_{i\,j} + \Delta(\mathcal F, \X)_{j\,k} = \Delta(\mathcal F, \X)_{i\,k},\;\; i,j,k=1,\dots s+1.
\end{equation}
\end{lemma}
\begin{proof}
Follows immediately from the definition.
\end{proof}

\begin{lemma}\label{lemma:PartitCoeff}
Let $\X\subseteq \mathbb{C}^p$ be a finite subset, $\mathcal F$ be a flag of length $s$ on $\mathbb{C}^p$, $\r$ be a sequence as in (\ref{seq_r}), and $\Delta(\mathcal F, \X)_{1\,i} < 0$ for all $i=2,\dots, s+1$. Then, there is a subflag $\mathcal F'\subseteq \mathcal F$ and a subsequence $\r'\subseteq \r$, both of length $t\leqslant s$ such that
\begin{equation}
S(\mathcal F,\X,\r)\leqslant S(\mathcal F',\X,\r'),
\end{equation}
\begin{equation}
\Delta(\mathcal F',\X)_{1\,2}<0,\;\; \text{and}\;\;\Delta(\mathcal F',\X)_{i\,i+1} \leqslant 0,\;i=2,\dots,t.
\end{equation}
In particular, $S(\mathcal F,\X,\r)< 0$.
\end{lemma}
\begin{proof}
The proof is by induction on $s$. For $s=1$,
\begin{equation}
S(\mathcal F,\X,\r) = r_1 \Delta(\mathcal F,\X)_{1\,2}<0,
\end{equation}
Let now $s>1$. If for all $i=1,\dots, s,\; \Delta(\mathcal F,\X)_{i\,i+1}\leqslant 0$, then we are done since $\Delta(\mathcal F,\X)_{1\,2} < 0$. Hence, we may assume that there is $i \leqslant s$ such that 
\begin{equation}
\Delta(\mathcal F,\X)_{j\,j+1}\leqslant 0,\; 1\leqslant j<i,\;\; \text{and}\;\; \Delta(\mathcal F,\X)_{i\,i+1} > 0,
\end{equation}
Let $\mathcal F'$ to be $\mathcal F$ without $V_i$ and $\r'$ to be $\r$ without $r_i$, then, 
\begin{multline}
S(\mathcal F,\X,\r) = \sum_{\substack{1\leqslant j\leqslant s\\j\neq i-1,i}}r_j \Delta(\mathcal F,\X)_{j\,j+1} + r_{i-1}\Delta(\mathcal F,\X)_{i-1\,i} \\ + r_i \Delta(\mathcal F,\X)_{i\,i+1} \\ \leqslant  \sum_{\substack{1\leqslant j\leqslant s\\j\neq i-1,i}}r_j \Delta(\mathcal F,\X)_{j\,j+1} + r_{i-1}\Delta(\mathcal F,\X)_{i-1\,i} + r_{i-1} \Delta(\mathcal F,\X)_{i\,i+1} \\ = S(\mathcal F',\X,\r'),
\end{multline}
where in the last equality we use Lemma \ref{lemma:SAdditiv}. Since the length of $\mathcal F'$ is less than that of $\mathcal F$ and $\Delta(\mathcal F',\X)_{1\,j}$ is either $\Delta(\mathcal F,\X)_{1\,j}$ or $\Delta(\mathcal F,\X)_{1\,j+1}$, thus strictly negative, the result follows by induction.
\end{proof}

\begin{proof}[Proof of Lemma \ref{geom_lem}]
Suppose on the contrary, that there exists a sequence $\{\bm\Theta_k\}_k \subset \mathcal{M}^\mathcal{G}$, such that $F(\bm\Theta_k)$ is bounded and
\begin{equation}
\bm\Theta_k \to \widetilde{\bm\Theta} \in \partial \mathcal{M}^\mathcal{G},
\end{equation}
meaning that $\rank{\widetilde{\bm\Theta}} < p$. The spectral decomposition of $\bm\Theta_k$ reads as
\begin{equation}
\bm\Theta_k = \sum_{j=1}^p \lambda_j \y_j\y_j^H,
\end{equation}
where $\y_j$ are orthonormal and $\lambda_j$ are all positive (we drop the dependence on $k$ to shorten the notation). Taking a subsequence of $\{\bm\Theta_k\}_k$, if needed, we may assume that all $\y_j$ and $\lambda_j$ converge to $\tilde\y_j$ and $\tilde \lambda_j$ correspondingly and $\tilde\y_j$ together with $\tilde \lambda_j$ determine a spectral decomposition of the limit $\widetilde{\bm\Theta}$. We will not mention explicitly taking subsequence argument but it is assumed that we do it when it is needed.

Note that some of the eigenvalues $\lambda_j$ tend to zero. Sort them according to their speed of convergence to zero, starting from $\lambda_1$, tending to zero most rapidly. Let $\mu$ be a sequence and $1\leqslant d <p$ be an integer such that $\frac{\log \lambda_j}{\log \mu}\to r_j > 0$ if $j\leqslant d$ and $\frac{\log \lambda_j}{\log \mu} \to 0$ if $j>d$ (such a sequence exists if we pass to a subsequence). Denote the set $\{\lambda_1,\ldots,\lambda_d\}$ by $\mathcal{L}$ and decompose it into a disjoint union $\mathcal{L} = \sqcup_{j=1}^m \mathcal{L}_j$ in such a way that there exist sequences $\mu_j,
\; i=1,\dots, m$ such that $\mu_j \to 0,\;\frac{\mu_j}{\mu_{j+1}}\to 0$, and, for each $\lambda \in \mathcal{L}_j,\;\frac{\lambda}{\mu_j}$ tends to a nonzero constant. The later means that for any $\lambda \in \mathcal{L}_j$, $\log \lambda \asymp r_j\log\mu$. For each $j=1,\dots,m$, define $K_j$ to be the random subspace generated by the eigenvectors corresponding to $\lambda$-s in $\mathcal{L}_j$. These are $\mathcal G$-invariant subspaces. We set $V_j$ to be $(\oplus_{l=1}^{j-1} K_l)^\perp,\; j=1,\dots,m+1$. Then $\mathcal F = \{V_j\}_{j=1}^{m+1}$ is a random flag of $\mathcal G$-invariant subspaces. 

Recall, that we need to prove that
\begin{equation}
F(\bm\Theta_k) = \log |\bm\Theta_k| + \frac{p}{n} \sum_{i=1}^n \log\(\x_i^H\bm\Theta_k^{-1}\x_i\) = \text{I} + \text{II},
\label{fd}
\end{equation}
tends to $+\infty$. Let us extract the main asymptotic term from $\text{I}$
\begin{multline}
\text{I} = \log |\bm\Theta_k| \asymp \sum_{j=1}^m\sum_{\lambda\in \mathcal{L}_j}\log \lambda \asymp \sum_{j=1}^m r_j|\mathcal{L}_j|\log  \mu \\ = \sum_{j=1}^m r_i(\dim V_j - \dim V_{j+1})\log  \mu.
\end{multline}
For $\text{II}$ we have
\begin{multline}
\text{II} = \frac{p}{n} \sum_{i=1}^n \log\(\x_i^H\bm\Theta_k^{-1}\x_i\) = \frac{p}{n}\sum_{j=1}^m \sum_{\x\in V_j\setminus V_{j+1}} \log \(\x^H\bm\Theta_k^{-1}\x\) \\
+ \frac{p}{n}\sum_{j=1}^m \sum_{\x\in V_{m+1}} \log \(\x^H\bm\Theta_k^{-1}\x\).
\end{multline}
Note that for any $\x \neq 0$,
\begin{equation}
\log \(\x^H \bm\Theta_k^{-1}\x\) \geqslant \log \frac{\norm{\x}^2}{\norm{\bm\Theta_k}_2} \geqslant \log\frac{\norm{\x}^2}{\Tr{\bm\Theta_k}} = 2\log\norm{\x} > -\infty,
\end{equation}
hence, we may ignore the summands with $\x\in V_{m+1}$ since they only improve the asymptotic. 

Now suppose that we are given $\lambda \in \mathcal{L}_j$ and consider the sequence $\lambda\bm\Theta_k^{-1}$. This sequence diverges, however, the limit of the sequence of the restricted operators $\lambda\bm\Theta_k^{-1}|_{V_j}$ exists and will be denoted by $\P_j$. Clearly, $\P_j$ is a composition of the orthogonal projector onto $K_j$ and a positive operator on $K_j$. Therefore, for any $\x\in V_j\setminus V_{j+1}$,
 \begin{equation}
\x^H\lambda\bm\Theta_k^{-1}\x \to \x^H \P_j \x> 0.
 \end{equation}
We obtain that asymptotically $\text{II}$ is not less than
\begin{align}
&\frac{p}{n}\sum_{j=1}^m \sum_{\x\in V_j\setminus V_{j+1}} \log \(\x^H\bm\Theta_k^{-1}\x\) \\
&=  \frac{p}{n}\sum_{j=1}^m \sum_{\x\in V_j\setminus V_{j+1}}\left(-\log\lambda +  \log \(\x^H\lambda\bm\Theta_k^{-1}\x\)\right) \nonumber
\\
&\asymp -\frac{p}{n}\sum_{j=1}^m \sum_{\x\in V_j\setminus V_{j+1}}\log \mu_{j} \asymp -\frac{p}{n}\sum_{j=1}^m r_j\left(|\X\cap V_{j}| - |\X\cap V_{j+1}|\right) \log \mu, \nonumber
\end{align}
here $\lambda$ belongs to the corresponding $\mathcal{L}_j$. Therefore, if the following expression is not zero, the leading term of $F$ is not smaller asymptotically than
\begin{multline}
\sum_{j=1}^m r_j\(\dim V_{j} - \dim V_{j+1} - \frac{p}{n}\(|\X\cap V_{j}| - |\X\cap V_{j+1}|\)\) \log\mu \\ =  S(\mathcal F,\X,\r)\log\mu.
\end{multline}
Note that
\begin{multline}
\Delta(\mathcal F,\X)_{1\,j} =
\dim V_1 - \dim V_j - \frac{\dim \mathbb{C}^p}{|\X|}(|\X\cap V_1| - |\X\cap V_j|
\\
=\dim V_1 \left(\frac{|\X \cap V_j|}{|\X|} - \frac{\dim V_j}{\dim \mathbb{C}^p}\right)<0,
\end{multline}
now Lemma \ref{lemma:PartitCoeff} yields that $S(\mathcal F,\X,\r)<0$ and $F(\bm\Theta_k)$ tends to infinity. This contradiction with the choice of $\bm\Theta_k$ completes the proof.
\end{proof}

\section{}
\label{lem3_app}
\begin{lemma}\label{lemma:intersectBound}
Let $L \subset \mathbb{C}^p$ be a random proper $\mathcal{G}$-invariant subspace and $d_i = \rank{\Pi_i}/p_i$, where $\Pi_i$ is the $i$-th block of matrix $\Pi_{L}$, as in (\ref{shur_f}), then for any continuously distributed independent $\X \subset \mathbb{C}^p$
\begin{equation}
|\X\cap \mathcal {L}| \leqslant \min\left\{ \left.\frac{d_i}{p_i}\;\right|\; d_i<s_i \right\},\;\; \text{a.s.}
\end{equation}
\end{lemma}
\begin{proof}
Let $\Y = \X\cap L$, denote $t = |\Y|$ and arrange the indices in such a way that $d_i<s_i$ for $i = 1,\dots,k$ and $d_i = s_i$ for $i>k$. Then condition $\sspan{\mathcal{G}\Y}\subset L$  implies that $\rank{\Pi_{\sspan{\mathcal{G}\Y},i}} \leqslant \rank{\Pi_i}$ for each $i$. For $i>k$, this condition is trivial because $\Pi_i$ is of the maximal rank. Corollary \ref{lemma:isotypicBound} yields
\begin{equation}
\rank{\Pi_{\sspan{\mathcal{G}\Y},i}} = t p_i^2 \leqslant p_id_i = \rank{\Pi_i}<s_ip_i,\;\; \text{a.s.}\;\forall \;i\leqslant k.
\end{equation}
Hence,
\begin{equation}
|\X\cap L| = |\Y| = t \leqslant \frac{d_i}{p_i},\;\; \text{a.s.}\;\forall\;i\leqslant k.
\end{equation}
\end{proof}

\begin{proof}[Proof of Lemma \ref{lemma:asInequality}]
Let $\Pi_i$ and $d_i$ be as in Lemma \ref{lemma:intersectBound}, and arrange the indices so that $d_i < s_i$ for $i\leqslant k$ and $d_i = s_i$ for $i>k$. From Lemma \ref{lemma:intersectBound},
\begin{equation}
\frac{|\X \cap L|}{n} \leqslant \frac{\min_{i=1}^k\frac{d_i}{p_i}}{n},\;\;\text{a.s.}
\end{equation}
Now it is enough to show that
\begin{equation}
\frac{\min_{i=1}^k\frac{d_i}{p_i}}{n} < \frac{\dim L}{p},\;\; \text{a.s.}
\label{f_1}
\end{equation}
Note that
\begin{equation}
\dim L = \sum_{i=1}^k p_id_i+ \sum_{i=k+1}^m p_is_i,
\end{equation}
therefore, (\ref{f_1}) is equivalent to
\begin{equation}
\frac{p \min_{i=1}^k\frac{d_i}{p_i}}{\sum_{i=1}^k p_id_i+ \sum_{i=k+1}^m p_is_i} < n,\;\; \text{a.s.}
\end{equation}
Due to
\begin{equation}
\min_{1\leqslant i \leqslant k}\frac{d_i}{p_i} \leqslant \min_{1\leqslant i \leqslant k}\frac{d_i}{s_i}\max_{1\leqslant i \leqslant k}\frac{s_i}{p_i},
\end{equation}
\begin{equation}
\sum_{i=1}^k p_id_i \geqslant \min_{1\leqslant i \leqslant k}\frac{d_i}{s_i}\sum_{i=1}^k p_is_i,
\end{equation}
it is enough to show that
\begin{equation}
\frac{p\min_{i=1}^k\frac{d_i}{s_i}\max_{i=1}^k\frac{s_i}{p_i}}{\min_{i=1}^k\frac{d_i}{s_i}\sum_{i=1}^k p_is_i + \sum_{i=k+1}^m p_is_i} < n,\;\; \text{a.s.}
\end{equation}
Since $0\leqslant d_i < s_i$, replace $\min_{i=1}^k d_i/s_i$ by $t < 1$ and we should prove that
\begin{equation}
\frac{p t \max_{i=1}^k\frac{s_i}{p_i}}{t\sum_{i=1}^k p_is_i + \sum_{i=k+1}^m p_is_i} < n,\;\; \text{a.s.}
\end{equation}
The left-hand side achieves its maximum at $t=1$, thus, with probability one, it is not greater than
\begin{equation}
\frac{p \max_{i=1}^k\frac{s_i}{p_i}}{\sum_{i=1}^m p_is_i} = \frac{p \max_{i=1}^k\frac{s_i}{p_i}}{p} = \max_{1\leqslant i \leqslant k}\frac{s_i}{p_i}.
\end{equation}
Since $n > s_i/p_i$ for any $i$, the statement follows.
\end{proof}

\section{}
\label{perf_app}
In this section For $n$ instances $a_1,\dots,a_n$ of scalars, vectors, matrices or functions, denote by $\wideparen{a}$ their arithmetic average, when the index of summation is obvious from the context.

\begin{lemma}(Vector Bernstein Inequality) \cite{yurinskiui1976exponential}
\label{v_bernstein_th}
Let $\bm\xi_1,\dots,\bm\xi_n \in \mathbb{C}^k$  be i.i.d zero-mean random vectors and suppose there exist $\sigma, \nu >0$ such that
\begin{equation}
\mathbb{E}\norm{\bm\xi_1}^r \leqslant \frac{r!}{2}\sigma^2 \nu^{r-2},\;\; r=2,3,\dots,
\end{equation}
then for $t\geq0$
\begin{equation}
\mathbb{P}\(\norm{\wideparen{\bm\xi}} \geqslant t \sigma\) \leq
2\exp\(\frac{-nt^2}{2(1+1.7t\frac{\nu}{\sigma})}\).
\end{equation}
\end{lemma}

For two matrices $\A$ and $\B$ of the same sizes, $\A \bullet \B$ denotes their Hadamard (elementwise) product. Given a group $\mathcal{G}$, we write $\M(\mathcal{G})$ for the mask matrix of the block-diagonal structure associated with it, having ones inside the corresponding blocks and zero otherwise.

\begin{lemma}
\label{u_mom}
Let $\x_i \sim \mathcal{U}(\I) , i=1,\dots,n$ then
\begin{equation}
\mathbb{E}\[\norm{\wideparen{\x\x^H}-\frac{1}{p}\I}_F^r\] \leqslant 1,\;\; r=2,3,\dots,
\label{xi_l_1}
\end{equation}
and for any group $\mathcal{G} \leqslant U(p)$,
\begin{equation}
\mathbb{E}\norm{\M(\mathcal{G})\bullet(\wideparen{\x\x^H}-\frac{1}{p}\I)}_F^r \leqslant \rho(\mathcal{G}).
\end{equation}
\begin{proof}
Define $n$ centered random vectors
\begin{equation}
\bm\xi_i = \vecc{\x_i\x_i^H-\frac{1}{p}\I},\;\; i=1,\dots,n,
\end{equation}
and consider their norm powers
\begin{multline}
\norm{\bm\xi_1}^r = \Tr{\(\x_1\x_1^H-\frac{1}{p}\I\)^2}^{\frac{r}{2}} \\
= \Tr{\(1-\frac{2}{p}\)\x_1\x_1^H+\frac{1}{p^2}\I}^{\frac{r}{2}} = \(1-\frac{1}{p}\)^{\frac{r}{2}} < 1,
\label{v_n_b}
\end{multline}
and (\ref{xi_l_1}) follows. Now define new $n$ centered random vectors
\begin{equation}
\bm\eta_i = \vecc{\M(\mathcal{G})\bullet\(\wideparen{\x\x^H}-\frac{1}{p}\I\)},\;\; i=1,\dots,n.
\end{equation}
Partition $\x$ as
\begin{equation}
\x = \begin{pmatrix}\x^1\\\vdots\\\x^l\end{pmatrix},
\end{equation}
according to the block-diagonal structure of $\M(\mathcal{G})$. Then
\begin{align}
&\norm{\M(\mathcal{G})\bullet\(\x\x^H-\frac{1}{p}\I_p\)}_F^2 = \sum_{l} \norm{\x^l[\x^l]^H-\frac{1}{p}\I_{s_l}}_F^2\nonumber\\
&=\sum_{l} \Tr{\(\x^l[\x^l]^H-\frac{1}{p}\I_{s_l}\)\(\x^l[\x^l]^H-\frac{1}{p}\I_{s_l}\)} \nonumber\\ &=\sum_{l} \(\norm{\x^l}^4-\frac{2}{p}\norm{\x^l}^2 + \frac{s_l}{p^2}\) =\(\sum_{l} \norm{\x^l}^4\)-\frac{2}{p}+\frac{1}{p} \nonumber\\
&=\(\sum_{l} \norm{\x^l}^2\)^2-2\sum_{l\neq k}\norm{\x^l}^2\norm{\x^k}^2-\frac{1}{p}\nonumber\\ &=1-2\sum_{l\neq k}\norm{\x^l}^2\norm{\x^k}^2-\frac{1}{p}.
\label{mask_e}
\end{align}
In order to calculate the expectation of (\ref{mask_e}) we only need to compute the subvectors' $\x^l$ and $\x^k$ norms moments. This is easily done by viewing the $p$ dimensional unit complex sphere as a $2p$ dimensional real sphere, with all the (sub)vectors of double dimensions. Apply the following formula from \cite{joarder2008distributions}:
\begin{equation}
\mathbb{E}\(\prod_{m=1}^d y_m^{k_m}\) = \frac{\Gamma(d/2)}{2^k\Gamma((d+k)/2)}\prod_{m=1}^d\frac{k_m!}{(k_m/2)!},
\end{equation}
where $\y=(y_1,\dots,y_d)^T$ is a real unit vector, all $k_m$ are even and $k=\sum_{m=1}^d k_m$ to obtain
\begin{equation}
\mathbb{E}\[\norm{\x^l}^2\norm{\x^j}^2\] = \sum_{a=1}^{2s_l}\sum_{b=1}^{2s_k} \mathbb{E}[y_a^2y_b^2],
\end{equation}
where $\y=[\operatorname{Re}(\x)^T, \operatorname{Im}(\x)^T]^T \in \mathbb{R}^{2p}$.
\begin{equation}
\mathbb{E}\[\norm{\x^l}^2\norm{\x^j}^2\] = 4^2s_ls_k \frac{(p-1)!}{2^4(p+1)!} = \frac{s_ls_k}{p(p+1)}.
\end{equation}
Now the expectation of (\ref{mask_e}) reads as
\begin{align}
&\mathbb{\E}\[\norm{\M(\mathcal{G})\bullet\(\x\x^H-\frac{1}{p}\I_p\)}_F^2\] = 1-\frac{1}{p}-2\sum_{l\neq j} \frac{s_ls_j}{p(p+1)} \nonumber\\
&= 1-\frac{1}{p}- \frac{p^2}{p(p+1)}(1-\rho(\mathcal{G})) =\frac{p^2}{p(p+1)}\rho(\mathcal{G})+\frac{1}{p+1} -\frac{1}{p} \nonumber\\
&\leqslant \rho(\mathcal{G}).
\end{align}
Note that
\begin{equation*}
\norm{\M(\mathcal{G})\bullet\(\x\x^H-\frac{1}{p}\I_p\)}_F \leqslant \norm{\x\x^H-\frac{1}{p}\I_p}_F < 1,
\end{equation*}
due to (\ref{v_n_b}), thus for $r \geqslant 3$,
\begin{align}
&\mathbb{\E}\[\norm{\M(\mathcal{G})\bullet\(\x\x^H-\frac{1}{p}\I_p\)}_F^r\] \nonumber\\
&= \mathbb{\E}\[\norm{\M(\mathcal{G})\bullet\(\x\x^H-\frac{1}{p}\I_p\)}_F^2\norm{\M(\mathcal{G})\bullet\(\x\x^H-\frac{1}{p}\I_p\)}_F^{r-2}\] \nonumber\\
&\leqslant \mathbb{\E}\[\norm{\M(\mathcal{G})\bullet\(\x\x^H-\frac{1}{p}\I_p\)}_F^2\] \leqslant \rho(\mathcal{G}).
\end{align}
\end{proof}
\end{lemma}

\begin{lemma}
\label{main_lem}
Let $\Z_1,\dots,\bm\Z_n \in \mathcal{S}(p)$  be i.i.d zero-mean random matrices and suppose there exist $\sigma, \nu>0$ such that
\begin{equation}
\mathbb{E}\norm{\Z_1}_F^r \leqslant \frac{r!}{2}\sigma^2 \nu^{r-2}, r=2,3,\dots,
\end{equation}
then for $t\geq0$
\begin{equation}
\mathbb{P}\(\norm{\wideparen{\Z}}_F \geqslant t \sigma\) \leq
2\exp\(\frac{-nt^2}{2(1+1.7t\frac{\nu}{\sigma})}\).
\end{equation}
If in addition $\U \in \mathcal{S}(p)$ possesses a sparsity pattern $\M$ and
\begin{equation}
\mathbb{E}\norm{\M\bullet\Z_1}_F^r \leqslant \alpha \mathbb{E}\norm{\Z_1}_F^r, r=2,3\dots,
\end{equation}
then
\begin{equation}
\mathbb{P}\(|\Tr{\U\wideparen{\Z}}| \geqslant t\sqrt{\alpha}\sigma\norm{\U}_F\) \leqslant 2\exp\(\frac{-nt^2}{2(1+1.7t\frac{\nu}{\sqrt{\alpha}\sigma})}\).
\end{equation}
\end{lemma}
\begin{proof}
The statement follows from the Cauchy-Schwartz inequality and Lemma \ref{v_bernstein_th} since
\begin{equation}
\Tr{\U\wideparen{\Z}} = \Tr{\U(\M\bullet\wideparen{\Z})} \leqslant \norm{\M\bullet\wideparen{\Z}}_F \norm{\U}_F.
\end{equation}
\end{proof}
\begin{corollary}
\label{dim_r_cor}
Let $\U \in \mathcal{S}(p)^\mathcal{G}$ and $\Z_i = p\(\x_i\x_i^H - \frac{1}{p}\I\)$, where $\x_i \sim \mathcal{U}(\I)$, then
\begin{equation}
\mathbb{P}\(|\Tr{\U\wideparen{\Z}}| \geqslant  tp\sqrt{\rho(\mathcal{G})}\norm{\U}_F\) \leqslant 2\exp\(\frac{-nt^2}{2(1+\frac{1.7t}{\sqrt{\rho(\mathcal{G})}})}\).
\end{equation}
\end{corollary}
\begin{proof}
Let $\Q$ be the orthogonal basis bringing all the matrices in $\mathcal{S}(p)^\mathcal{G}$ to the block-diagonal form with the mask $\M(\mathcal{G})$, according to Theorem \ref{rep_basis}, then
\begin{equation}
\Tr{\U\wideparen{\Z}} = \Tr{\Q^H\U\Q\Q^H\wideparen{\Z}\Q} \leqslant \norm{\U}_F\norm{\M(\mathcal{G})\bullet\Q^H\wideparen{\Z}\Q}_F,
\end{equation}
where we have used the equality $\norm{\Q^H\U\Q}_F = \norm{\U}_F$. Use Lemmas \ref{main_lem} and \ref{u_mom} with $\alpha = \rho(\mathcal{G}),\; \sigma = \nu = 1$,  and note that the distribution of $\Q^H \wideparen{\Z} \Q$ is identical to that of $\wideparen{\Z}$ to get the statement.
\end{proof}

\bibliographystyle{IEEEtran}
\bibliography{ilya_bib}

\end{document}